\documentclass{llncs}
\usepackage{float}
\usepackage{graphicx}
\usepackage{color}
\usepackage{tikz}
\usepackage{amsmath}
\usepackage{amssymb}
\usepackage[english]{babel}
\usepackage{flafter}
\usepackage{array}
\usepackage{paralist}
\usepackage{algorithm}
\usepackage[noend]{algorithmic}
\usepackage[toc, page]{appendix}
\input{mathdots}

\DeclareMathOperator*{\argmin}{arg\,min}
\newcommand{\om}{\omega}

\begin{document}

% % \begin{center}
% % \LARGE The Online Replacement Path Problem
% % \end{center}

\title{The Online Replacement Path Problem}

\author{David Adjiashvili\inst{1} \and Marco Senatore\inst{2}}

\institute{ Institute for Operations Research (IFOR) \\
 Eidgen\"ossische Technische Hochschule (ETH) Z\"urich \\
 R\"amistrasse 101, 8092 Z\"urich, Switzerland \\
 \email{david.adjiashvili@ifor.math.ethz.ch} 
\and
 Dipartimento di Informatica, Sistemi e Produzione \\
 Universit\`a di Roma “Tor Vergata“ \\
 Via del Politecnico 1, 00133 Rome, Italy \\
 \email{senatore@disp.uniroma2.it}
}

\maketitle

\begin{abstract}
We study a natural online variant of the replacement path problem. The \textit{replacement path problem} asks to find
for a given graph $G = (V,E)$, two designated vertices $s,t\in V$ and a shortest $s$-$t$ path $P$ in $G$, a
\textit{replacement path} $P_e$ for every edge $e$ on the path $P$. The replacement path $P_e$ is simply a shortest
$s$-$t$ path in the graph, which avoids the \textit{failed} edge $e$. We adapt this problem to deal with the natural
scenario, that the edge which failed is not known at the time of solution implementation. Instead, our problem
assumes that the identity of the failed edge only becomes available when the routing mechanism tries to cross the
edge. This situation is motivated by applications in distributed networks, where information about recent changes in the
network is only stored locally, and fault-tolerant optimization, where an adversary tries to delay the
discovery of the materialized scenario as much as possible. Consequently,
we define the \textit{online replacement path problem}, which asks to find a nominal $s$-$t$ path $Q$ and detours
$Q_e$ for every edge on the path $Q$, such that the worst-case arrival time at the destination is minimized. Our main
contribution is a label setting algorithm, which solves the problem in undirected graphs
in time $O(m \log n)$ and linear space for all sources and a single destination.
We also present algorithms for extensions of the model to any bounded number of failed edges. 

\end{abstract}

% \newpage

\section{Introduction}\label{sec:inro}

Modeling the effects of limited reliability of networks in modern routing schemes is important from the 
point of view of most applications. It is often unrealistic to assume that the nominal network known at 
the stage of decision making will be available in its entirety at the stage of solution implementation. 
Several research directions have emerged as a result. The main paradigm in most works is to obtain a certain 
'fault-tolerant' or 'redundant' solution, which takes into account a certain set of likely network realizations at the 
implementation phase. One important example is the \textit{replacement path problem} (RP) \cite{RepPathFirst,RepPathUndir}. 
The input in RP 
is a nominal network given as a graph $G=(V,E)$, a source $s$, a destination $t$ and one shortest $s$-$t$ path
 $P$. The goal is to find for every edge $e$ on the path $P$, a shortest path in $G$, which does not use the
 edge $e$. RP attempts to model the situation in which any link in the network may fail before the routing 
process starts. 
It is hence desirable to compute in advance the shortest replacement paths, for the case of a failure
of any one of the edges in the nominal path $P$.
% which avoid every single edge in the nominal path $P$. 
In the event of a failure the routing mechanism simply chooses the corresponding 
pre-computed path. The applicability of RP is, however, limited to those situations, in which it is possible 
to know the identity of the failed link before the routing process starts. This assumption 
is not realistic in many important applications, in which faults in the network occur 'online', or the information
 about them is stored in a distributed fashion. The latter situation is commonplace, for example, in 
transportation networks (e.g. accidents in road networks). Furthermore, it is a common feature of very 
large networks, such as the Internet. This paper studies the \textit{online replacement path problem} 
(ORP), which captures this online failure setting. The most notable difference between RP and ORP is that 
in ORP we assume that the routing mechanism is informed about the failed link at the moment it tries to use
it. Another important difference is related to the nominal path $P$. In RP a certain nominal shortest
path is provided in the input. This is no longer the case in ORP, since the detour taken in the event of
a failure does not always start from the source vertex $s$. This means that in ORP we simultaneously 
optimize both the nominal path and the optimal detours, taking into account a certain global objective function.
%  that ORP tries also tries to optimize it might not be optimal anymore to use a shortest path as the nominal path.
An informal formulation of ORP is as follows. We would
like to route a certain package through a network from a given source to a given destination as quickly
as possible. We are aware of the existence of a failed link in the network, but we do not know its
location. It is possible to observe that a certain link has failed by \textit{probing} it. In order to probe a
link, the package should be at one of the endpoints of this link. If a probed link is intact, the package
crosses the link to the other endpoint and the cost of traversing the link in incurred. 
Otherwise the package stays in the same endpoint
and the routing mechanism is informed about the failed link. In other words, it is only possible to observe that
a link has failed by trying to cross it.
The goal is to find a set of paths (a nominal path and detours for every edge on the nominal path)
that will minimize the latest possible arrival time to the destination. 
A solution of ORP should hence specify both the nominal path and the optimal detours at every vertex along the path, 
which avoid the next edge on the nominal path. 
% \footnote{
% It is also possible to define a variant in which a nominal path is given in the input, and the
% information about the failure in the network is stored locally. The algorithms presented here are adaptable
% to this setting as well with the same complexity.}

The latter informal definition suggests that in ORP we take a conservative fault-tolerant approach.
In fact, it is assumed that a failed link does exist in the network, but its identity is unknown.
This suggests another application of ORP. In many applications it is only necessary to route a certain
object within a certain time, called a \textit{deadline}. As long as the object reaches its destination before the
deadline, no penalty is incurred. On the other hand, if the deadline is not met, a large penalty is due.
An example of such an application is organ transportation for transplants 
(see e.g. Moreno, Valls and Ribes~\cite{OrganTrans}), in which it is critical to deliver
a certain organ before the scheduled time for the surgery. In this application it does not matter how early the organ
arrives at the destination, as long as it arrives in time. In such applications it is
often too risky to take an unreliable shortest path, which admits only long detours in some scenarios, 
whereas a slightly longer path with reasonably short detours meets the deadline in \textit{every} scenario.
Hence, with ORP it is often possible to immunize the path against faults in the network.

The main result of this paper is that the solution to ORP in undirected networks can be computed in 
$O(m \log n)$ time and linear space for all sources and a single destination, where $n$ and $m$ are the number of 
vertices and edges in the network, respectively. Furthermore, this solution can be stored in $O(n)$ space.
% This running time is asymptotically the same as that of computing a single ordinary shortest path. 
The basic algorithm is a label-setting algorithm, similar to Dijkstra's algorithm
for ordinary shortest paths. We describe this algorithm in Section~\ref{sec:k_is_one}.
The main technical difficulty lies in the need to pre-compute certain
shortest path distances in modified graphs. This difficulty is overcome in Section~\ref{sec:impl}, which provides
a fast implementation of this step. We generalize the model to incorporate the possibility of
an arbitrary bounded number $k$ of failed edges in Section~\ref{sec:korp}. This section also gives an alternative 
view on the problem using the notion of \textit{routing strategies}, and provides a polynomial algorithm
for the problem. We mention other results linking ORP with the shortest path problem in Section~\ref{sec:ORPvsSP}.
In particular, we show that it is possible to solve a bi-objective variant of ORP.
This result implies that it is possible to efficiently obtain 
Pareto-optimal paths with respect to ordinary distance and the cost corresponding to the ORP problem,
making it an attractive method for various applications. We also analyze the performance of a greedy 
heuristic which attempts to route along a shortest path in the remaining network. We show that this
heuristic, which is implemented in many applications, is a poor approximation for ORP.   
We summarize in Section~\ref{sec:conclusions}. The following section reviews related
work.

\section{Related Work}\label{sec:relatedwork}

The replacement path problem was first introduced by Nisan and Ronen~\cite{RepPathFirst}. The motivation
for their definition stemmed from the following question in auction theory: what is the true price of a link
in a network, when we try to connect two distinct vertices $x$ and $y$, and every edge is owned by a
self-interested agent. It turns out that compensating agents with respect to the declarations of other agents
in the auction leads to truthful declarations, namely declarations which reflect the true costs of the agent.
Such pricing schemes are called \textit{Vickrey schemes}. In the setup of networks, replacement paths lengths
correspond to Vickrey prices for the individual failed edges.
Another important application of RP is the \textit{$k$ shortest simple paths problem} (kSSP), which reduces to
$k$ replacement path computations. 

The complexity of the RP problem for undirected graphs is well understood. 
Malik, Mittal and Gupta~\cite{FirstUndirRP} give a simple $O(m + n \log n)$ algorithm. A mistake
in this paper was later corrected by Bar-Noy, Khuller and Schieber~\cite{FirstUndirRPFix}.
This running time
is asymptotically the same as a single source shortest path computation. 
Nardelli, Proietti and Widmayer~\cite{NodeRPUndir} later provided an algorithm with the
same complexity for the variant of RP, in which vertices are removed instead of edges.
The same authors give efficient algorithms for finding detour-critical edges for a given shortest path
in~\cite{DetourCritical1,DetourCritical2}.
% For an edge $e=uv$ on a given the shortest path, the detour cost is defined as the difference between 
% the length of the remaining path from $e$ and the length of the shortest detour from $u$ 
% to the destination, avoiding $e$.  

In directed graphs the situation is significantly different. 
A trivial upper bound for RP corresponds to $O(n)$ single shortest path computations. This
gives $O(n(m + n \log n))$ for general directed graphs with nonnegative weights. This was slightly
improved to $O(mn + n^2 \log \log n)$ by Gotthilf and Lewenstein~\cite{GenRepPathImpr}. The challenge
of improving the $O(mn)$ bound for RP on directed graphs was mainly tackled by restricting the
class of graphs or by allowing approximate solutions. Along the lines of the former approach, algorithms
were developed for unweighted graphs (Roditty and Zwick~\cite{DirUnweightedRP}) and planar graphs 
(Emek, Peleg and Roditty~\cite{DirRPPlan1}, Klein, Mozes, and Weimann~\cite{DirRPPlan2} and 
Wulff-Nilsen~\cite{DirRPPlan3}). The latter approach was successfully applied to obtain $\frac{3}{2}$-approximate
solutions by Roditty~\cite{3over2RP} and $(1+\epsilon)$-approximate solutions by Bernstein~\cite{1plusEpsRP}.
Weimann and Yuster~\cite{RPMatMul} applied fast matrix multiplication techniques to obtain a randomized algorithm with
sub-cubic running time for certain ranges of the edge weights.

Another problem which bears resemblance to ORP is the \textit{stochastic shortest path with recourse problem} (SSPR),
studied by Andreatta and Romeo~\cite{StochSPRecourse}. 
This problem can be seen as the stochastic analogue of ORP.
Finally, we briefly review some related work on robust counterparts of the shortest path problem.
The shortest path problem with cost uncertainty was studied by Yu and Yang~\cite{RSP}, who
consider several models for the scenario set. These results were later extended by
Aissi, Bazgan and Vanderpooten~\cite{MinMaxRegretSP}.
These works also considered a two-stage min-max regret criterion.
Dhamdhere, Goyal, Ravi and Singh~\cite{DemandR} developed the demand-robust model and gave an approximation
algorithm for the shortest path problem. A two-stage feasibility counterpart of the shortest path problem 
was addressed in Adjiashvili and Zenklusen~\cite{AdaptRob}. Puhl~\cite{RRSP} provided hardness results for 
numerous two-stage counterparts the shortest path problem, and gave some approximation algorithms.

\section{An Algorithm for ORP}\label{sec:k_is_one}

In this section we develop an algorithm for ORP. Some technical proofs are left in Appendix~\ref{app:proofs}.
Let us establish some notations first. 
We are given an undirected edge-weighted graph $G=(V,E,\ell)$, a source $s\in V$ and destination $t\in V$. 
We are assuming throughout this paper that the edge weights $\ell$ are nonnegative.
For two vertices $u,v \in V$ let $\mathcal{P}_{u,v}$ denote the set of simple $u$-$v$ paths in $G$. 
Let $N(u)$ denote the set of neighbors of $u$ in $G$. For a set of edges 
$A \subset E$ let $\ell(A) = \sum_{e\in A} \ell(e)$. For an edge $e \in E$ and a set 
of edges $F \subseteq E$, let $G-e$ and $G-F$ denote the graph obtained by removing the edge $e$ and the edges in $F$, 
respectively. For a 
graph $H$ let $d_H(\cdot,\cdot)$ denote the shortest path distance in $H$. Paths are always represented 
as sets of edges, while walks are represented as sequences of vertices. For a path $P$ with incident vertices 
$u$ and $v$ let  $P[u,v]$ denote the subpath of $P$ from $u$ to $v$. For an edge $e\in E$ and $u\in V$ let
$$
s^{-e}_u = d_{G-e}(u,t)
$$
denote the shortest $u$-$t$ path distance in $G-e$.

Our algorithm uses a label-setting approach, analogous to Dijkstra's algorithm for shortest paths. In other
words, in every iteration the algorithm updates certain tentative labels for the vertices of the graph, and fixes
a final label to a single vertex $u$. This final label represents the connection cost of $u$ by an optimal path 
to $t$.

\begin{definition}\label{def:robustlenght}
Given a vertex $v \in V$, the \textit{robust length} of the $v$-$t$ path $P$ is 
$$
\mathrm{Val}(P) =  \max\{\ell(P), \max_{uu'\in P} \{\ell(P[v,u]) + s^{-uu'}_u\}\}.
$$
The potential $y(v)$ is defined as the minimum of $\mathrm{Val}(P)$ over all $P\in \mathcal{P}_{v,t}$, and
any path $P^*$ attaining $\mathrm{Val}(P^*) = y(v)$ is called an optimal nominal path. Finally, $ORP$ is to
compute $y(s)$ and obtain a corresponding optimal nominal path.
\end{definition}
The robust length of a $v$-$t$ path $P$ is simply the maximal possible cost incurred by following $P$ until a 
certain vertex, and then taking the best possible detour from that vertex to $t$ which avoids the next edge on 
the path. To avoid confusion, we stress that in ORP we assume the existence of at most one failed edge in the graph.
Consider next a scenario in which an edge $uu'\in P$ fails and let $u\in V$ be the vertex which is 
closer to $v$. We can assume without loss of generality that the best detour is 
a shortest $u$-$t$ path in the graph $G-uu'$. Critically, the values $s^{-uu'}_u$ are 
independent of the chosen path.

Observe that from non-negativity of $\ell$ we obtain $\mathrm{Val}(P) \geq \mathrm{Val}(P')$, whenever $P$
and $P'$ are $u$-$t$ and $v$-$t$ paths respectively, and $P'$ is a subpath of $P$.
We denote this property by \textit{monotonicity}. Furthermore, we can prove the following.

\begin{lemma}\label{lem:Tproperty}
Let $P_u \in \mathcal{P}_{u,t}$ and let $v \in N(u)$ be a vertex, not incident to $P_u$.
% for which $\max\{\ell(vu)+\mathrm{Val}(P_u),s^{-vu}_{v} \}$ is attained. 
Then the path $P_v = P_u \cup \{vu\}$ satisfies
\begin{equation*}
\mathrm{Val}(P_v)=\max\{\ell(vu)+\mathrm{Val}(P_u),s^{-vu}_{v} \}. 
\end{equation*} 
\end{lemma}
Our algorithm for ORP updates the potential on the vertices of the graph, using the property 
established by the following lemma.

\begin{lemma}\label{lem:dp_operator}
Let $U\subset V$, with $t \in U$, be the set of vertices for which the potential is known. And let $uv$ be the edge such that:
\begin{equation}\label{eq:DPequation}
uv = \argmin_{zw\in E: w\in U,z \in V \setminus U}\{\max\{\ell(zw)+y(w),s^{-zw}_{z}\}\}.
\end{equation}
Then if $P_u \in \mathcal{P}_{u,t}$ with $\mathrm{Val}(P_u) = y(u)$ and $P_v=P_u \cup \{uv\}$
it holds that $\mathrm{Val}(P_v) = y(v)$. 
\end{lemma}
Lemma~\ref{lem:dp_operator} provides the required equation for our label-setting algorithm, whose
formal statement is given as Algorithm~\ref{alg:main}. The algorithm iteratively builds up a set
$U$, consisting of all vertices, for which the correct potential value of $y(u)$ was already
computed. The correctness of the algorithm is a direct consequence of Lemma~\ref{lem:dp_operator}.

\begin{algorithm}
\caption{ }\label{alg:main}
\begin{algorithmic}[1]
\STATE{Compute $s^{-uv}_{u}$ for each $uv \in E$.}
\STATE $U = \emptyset$; \quad $W = V$; \quad $y'(t) = 0$; \quad $y'(u)=\infty \,\, \forall u\in V-t$.
\STATE{$successor(u)= \text{NIL} \,\, \forall u\in V$.}
\WHILE{$U \neq V$}
      \STATE{Find $u = \argmin_{z \in W}y'(z)$. }
      \STATE $U = U + u$; \quad $W = W - u$; \quad $y(u)=y'(u)$.
      \FORALL{ $vu \in E$ with $v \in W$}
      	\IF{ $y'(v) > \max\{\ell(vu)+y'(u),s^{-vu}_{v}\}$ }
      		\STATE{$y'(v) = \max\{\ell(vu)+y'(u),s^{-vu}_{v}\}$.}
		\STATE{$successor(v) = u$.}
	\ENDIF
      \ENDFOR
			
\ENDWHILE
\end{algorithmic}
\end{algorithm}
Consider the running time of Algorithm~\ref{alg:main}. We let $n$ and $m$ denote the number of vertices and edges
of the input graph, respectively. An efficient 
implementation of step~$1$ is delayed to the next section, and constitutes the heart of our efficient algorithms. 
For steps~$2$-$10$ we use the implementation of Fredman and Tarjan~\cite{FibHeap} for priority queues (heaps) called 
\textit{Fibonacci Heaps}. 
% A Fibonacci heap supports Insert and DecreaseKey operations in $O(1)$ amortized time and 
% ExtractMin operations in $O(\log n)$ amortized time. 
We adopt here the same implementation that is used to obtain $O(m + n \log n)$ running time for Dijkstra's algorithm. 
We omit the details as they are identical to those in~\cite{FibHeap}. We comment that our implementation of step~$1$
uses $O(m\log n)$ time, hence this computation dominates the running time. In fact, the running time of Algorithm~\ref{alg:main}
remains the same if steps~$2$-$10$ are implemented using simpler data structures, such as Binary Heaps.

Finally, let us remark that Algorithm~\ref{alg:main} works both for directed and undirected graphs. However,
the following section provides an implementation of step~$1$ for undirected graphs only. We remark that in 
directed graphs, Algorithm~\ref{alg:main} can be trivially implemented in time $O(n\mathrm{SP}(n,m))$, 
where $\mathrm{SP}(n,m)$ is the complexity of a single shortest path computation on a graph with $n$ vertices 
and $m$ edges. This is a simple consequence of our first observation in the following section.

\section{Efficient Computation of $s$-Values}\label{sec:impl}

It remains to provide an efficient implementation for the computation of the values $s^{-uu'}_u$.
% We provide two implementations of this step. 
We will assume a random access machine (RAM) as
a computational model. 
% In particular, we need indirect addressing (constant-time access to positions
% in arrays), but no bit-shifts.
% The following lemma provides a procedure that achieves an implementation in time $O(m \log n)$,
% and sets the ground for our main algorithm, which achieves a running time of $O(m + n \log n)$.

We start by computing the shortest path tree $T$
in $O(m \log n)$ time from every vertex to $t$. Let $d^*(u) = d_G(u,t)$.
Observe that $s^{-uu'}_u = d^*(u)$ holds for every edge $uu'$ outside of $T$.
We can hence concentrate our efforts on computing $s$-values for edges in the tree.
For a vertex $v\in V$ we denote by $T_v \subset T$ the subtree rooted at $v$. We set 
$E' = E\setminus T$.

\begin{lemma}\label{lem:expression_svalues}
Consider a vertex $u \in V$. Let $vw\in E'$ be an edge attaining
\begin{equation}\label{eq:computing_s1}
\min_{vw \in E', v\in T_u w\not\in T_u} \{d_G(u,v) + \ell(vw) + d^*(w)\},
\end{equation}
and let $uu'$ be the first edge on the $u$-$t$ path in $T$. Then 
$$
s^{-uu'}_u = d_G(u,v) + \ell(vw) + d^*(w).
$$
\end{lemma}
By defining 
$c_{vw} = d^*(v) + d^*(w) + \ell(vw)$
for each edge $vw\in E'$ and substituting in the expression for $s^{-uu'}_u$ obtained in the previous lemma we can write
\begin{equation}\label{eq:computing_s3}
s^{-uu'}_u = c_{vw} - d^*(u),
\end{equation}
which is a convenient expression for computing the $s$-values. Indeed this expression shows that value 
$s^{-uu'}_u$ only depends on the lowest value $c_e$, over all $e$ with exactly one incident vertex in $T_u$.

We leave the remaining details of our implementation to the proof of Theorem~\ref{thm:compexity_s}.
Before stating the theorem, let us define an important ingredient, which is used hereafter.
To efficiently obtain the information, whether a certain edge $e\in E'$ corresponds to a feasible detour for $u$ (namely if exactly
one endpoint of $e$ is in $T_u$), we need an algorithm for computing the \textit{least common ancestor (LCA)} in trees. Given
a tree rooted at a vertex $t$ and two vertices $u,v$ in the tree, the least common ancestor $\mathrm{lca}(u,v)$ of $u$ and $v$ is 
the vertex at which the $t$-$u$ and $t$-$v$ paths diverge in the tree.
% the farthest vertex from $t$, which appears on both the $u$-$t$ path and the $v$-$t$ path in the tree. 
% In other words, the LCA is the common ancestor of $u$ and $v$, which is the farthest from $t$. We will denote the LCA of two vertices $u,v$ in $T$ by $\mathrm{lca}(u,v)$. 
For an edge $e$ we write $\mathrm{lca}(e)$ to denote the LCA of the endpoints of $e$. 
It is straightforward to see that an edge $e = vw$ corresponds to a feasible detour for an edge $uu'$ if and only if
either the $v$-$\mathrm{lca}(v,w)$ path or the $w$-$\mathrm{lca}(v,w)$ path in $T$ contain $uu'$.
For this purpose we use the algorithm of Gabow and Tarjan~\cite{LCA},
which uses $O(n + p)$ time to compute the LCA of $p$ pairs of vertices in a rooted tree. In our case we have $p=O(m)$,
since we would like to compute the LCA for every pair of vertices connected by an edge $e\in E'$, hence this computation
takes $O(m)$ time. We assume henceforth that given an edge $e\in E'$ we have access to $\mathrm{lca}(e)$
in constant time. 
\begin{theorem}\label{thm:compexity_s}
The values $s^{-uu'}_u$ can be computed in $O(m \log n)$ time and linear space for all $uu' \in E$. 
Furthermore, they can be stored in a data structure of size $O(n)$.
% with constant query complexity.
\end{theorem}

\begin{proof}

The algorithm is summarized as Algorithm~\ref{alg:impl} in Appendix~\ref{svaluesalg}.
The algorithm starts by sorting the set of edges $E'$ according to increasing
order of $c_e$. Let $L$ be the sorted list. This operation (as well as the computation of the
values $c_e$) takes $O(m \log n)$ time.

The algorithm relies on the following fact. Let $e$ be the first edge in the sorted list $L$, which
represents a feasible detour for $uu' \in T$. From the previous discussion, the edge $e$ corresponds to the optimal detour for $uu'$. Furthermore,
if we traverse the two paths from the endpoints of $e$ to $\mathrm{lca}(e)$ we cross the edge $uu'$. The algorithm indeed
performs the latter traversals in a copy of $T$ and marks all edges along the way (which were not marked before, using an edge
that has lower $c$-value) as belonging to the edge $e$. To obtain the desired running time we need to avoid traversing the
same parts of $T$ several times. To achieve this we create a second copy $\bar T$ of the tree $T$. 
For every vertex $u\in T$ we store a pointer $p[u]$ from $u\in T$ to
its copy $\bar u$ in $\bar T$, and another pointer $p[\bar u]$, pointing in the opposite direction. 
We start iterating over the edges in the sorted list $L$. In the beginning of the $i$'th iteration, the first edge $e_i =xy$ in $L$
is removed from $L$, and the algorithm jumps to the corresponding vertices $\bar x = p[x]$ and $\bar y = p[y]$ in $\bar T$. 
The algorithm also memorizes $\bar z = \mathrm{lca}(\bar x,\bar y)$ - 
the vertex in $\bar T$, corresponding to $z = \mathrm{lca}(x,y)$. Next the $\bar x$-$\bar z$ and $\bar y$-$\bar z$ paths are traversed
in $\bar T$, marking all edges along the way as belonging to $e_i$, and computing the $s$-values for them using (\ref{eq:computing_s3}). 
Finally all these aforementioned edges are contracted
in $\bar T$ and the pointers are updated so that every vertex in $T$ always points to its corresponding
super-vertex in $\bar T$ and vice-versa. 
% This ends the $i$'th iteration. 
Finally, the list is post-processed
to remove some unnecessary edges in the following way. At every step, the first edge $e = vw$ in the list is inspected.
If $p[v] = p[w]$, this edges is a self-loop in $\bar T$, and it can not correspond to correct $s$-value assignments
anymore. Consequently, this edges is removed from the list and the next edge is inspected. The process ends when
$p[v] \neq p[w]$, or when $L$ is empty. In the latter case the algorithm terminates. This concludes the $i$'th
iteration.

Note that the running time of the $i$'th iteration is proportional to 
the combined length of the $\bar x$-$\bar z$ path, the $\bar y$-$\bar z$ path and the prefix of $L$ that is
deleted in the post-processing. Since we contract every edge we traverse,
in the tree $\bar T$, we never traverse the same edge twice. We conclude that the running time of this last stage of the
algorithm is $O(n + m)$. Figure~\ref{fig:treealg} illustrates the algorithm.

Finally, note that we can store a representation of all $s$-values and the corresponding detours in $O(n)$ space.
This is achieved by storing the tree $T$, and for every vertex $u$ and the corresponding edge $uu'$ in $T$, we
store a pointer to the edge attaining the minimum in (\ref{eq:computing_s1}). The values $s^{-uu'}_u$ can
either be stored in a separate list using $O(m)$ additional space, or computed in constant time from the aforementioned
information, by using (\ref{eq:computing_s3}).

% We conclude that the $s$-values can be computed in time $O(m \log n)$. Note that the running time is dominated by application
% of the sorting algorithm. 
\qed
\end{proof}

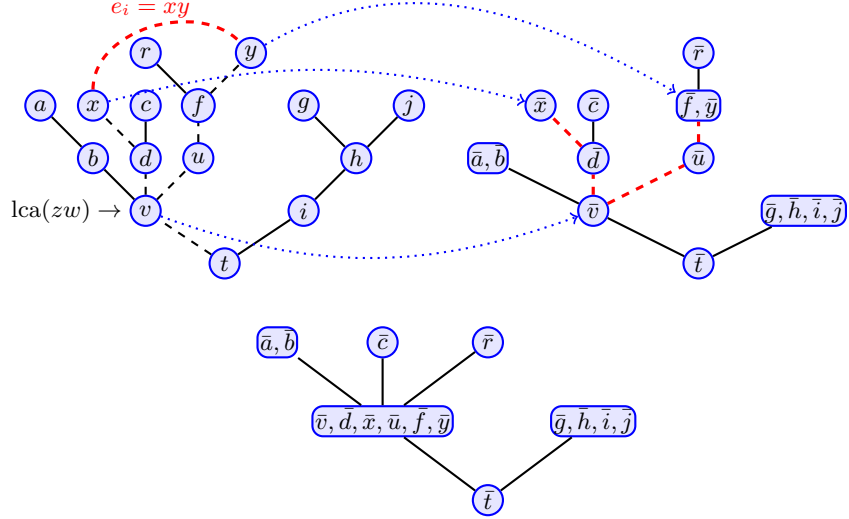
\begin{figure}[h]
\begin{center}
\begin{tikzpicture}[scale=0.7]

% \begin{footnotesize}

%main graph
\node (t) at (3,0) [circle,draw=blue!100,fill=blue!10,thick,inner sep=1pt,minimum size=4mm] {$t$};

\node (v) at (1.5,1) [circle,draw=blue!100,fill=blue!10,thick,inner sep=1pt,minimum size=4mm] {$v$};
\node (i) at (4.5,1) [circle,draw=blue!100,fill=blue!10,thick,inner sep=1pt,minimum size=4mm] {$i$};

\node (b) at (0.5,2) [circle,draw=blue!100,fill=blue!10,thick,inner sep=1pt,minimum size=4mm] {$b$};
\node (d) at (1.5,2) [circle,draw=blue!100,fill=blue!10,thick,inner sep=1pt,minimum size=4mm] {$d$};
\node (u) at (2.5,2) [circle,draw=blue!100,fill=blue!10,thick,inner sep=1pt,minimum size=4mm] {$u$};
\node (h) at (5.5,2) [circle,draw=blue!100,fill=blue!10,thick,inner sep=1pt,minimum size=4mm] {$h$};

\node (a) at (-0.5,3) [circle,draw=blue!100,fill=blue!10,thick,inner sep=1pt,minimum size=4mm] {$a$};
\node (w) at (0.5,3) [circle,draw=blue!100,fill=blue!10,thick,inner sep=1pt,minimum size=4mm] {$x$};
\node (c) at (1.5,3) [circle,draw=blue!100,fill=blue!10,thick,inner sep=1pt,minimum size=4mm] {$c$};
\node (f) at (2.5,3) [circle,draw=blue!100,fill=blue!10,thick,inner sep=1pt,minimum size=4mm] {$f$};
\node (g) at (4.5,3) [circle,draw=blue!100,fill=blue!10,thick,inner sep=1pt,minimum size=4mm] {$g$};
\node (j) at (6.5,3) [circle,draw=blue!100,fill=blue!10,thick,inner sep=1pt,minimum size=4mm] {$j$};

\node (e) at (1.5,4) [circle,draw=blue!100,fill=blue!10,thick,inner sep=1pt,minimum size=4mm] {$r$};
\node (z) at (3.5,4) [circle,draw=blue!100,fill=blue!10,thick,inner sep=1pt,minimum size=4mm] {$y$};

% edges

\draw [-,black,thick,dashed] (t) to node [auto] {} (v);
\draw [-,black,thick] (t) to node [auto] {} (i);

\draw [-,black,thick] (v) to node [auto] {} (b);
\draw [-,black, thick, dashed] (v) to node [auto] {} (d);
\draw [-,black, thick, dashed] (v) to node [auto] {} (u);
\draw [-,black,thick] (i) to node [auto] {} (h);

\draw [-,black,thick] (b) to node [auto] {} (a);
\draw [-,black, thick, dashed] (d) to node [auto] {} (w);
\draw [-,black,thick] (d) to node [auto] {} (c);
\draw [-,black, thick, dashed] (u) to node [auto] {} (f);
\draw [-,black,thick] (h) to node [auto] {} (g);
\draw [-,black,thick] (h) to node [auto] {} (j);

\draw [-,black,thick] (f) to node [auto] {} (e);
\draw [-,black, thick, dashed] (f) to node [auto] {} (z);

% T' before the ith iteration

\node (pt) at (12,0) [circle,draw=blue!100,fill=blue!10,thick,inner sep=1pt,minimum size=4mm] {$\bar t$};

\node (pv) at (10,1) [circle,draw=blue!100,fill=blue!10,thick,inner sep=1pt,minimum size=4mm] {$\bar v$};
\node (pghij) at (14,1) [rounded corners,draw=blue!100,fill=blue!10,thick,inner sep=1pt,minimum size=4mm] {$\bar g,\bar h,\bar i,\bar j$};

\node (pab) at (8,2) [rounded corners,draw=blue!100,fill=blue!10,thick,inner sep=1pt,minimum size=4mm] {$\bar a,\bar b$};
\node (pd) at (10,2) [circle,draw=blue!100,fill=blue!10,thick,inner sep=1pt,minimum size=4mm] {$\bar d$};
\node (pu) at (12,2) [circle,draw=blue!100,fill=blue!10,thick,inner sep=1pt,minimum size=4mm] {$\bar u$};

\node (pw) at (9,3) [circle,draw=blue!100,fill=blue!10,thick,inner sep=1pt,minimum size=4mm] {$\bar x$};
\node (pc) at (10,3) [circle,draw=blue!100,fill=blue!10,thick,inner sep=1pt,minimum size=4mm] {$\bar c$};
\node (pfz) at (12,3) [rounded corners,draw=blue!100,fill=blue!10,thick,inner sep=1pt,minimum size=4mm] {$\bar f,\bar y$};

\node (pe) at (12,4) [circle,draw=blue!100,fill=blue!10,thick,inner sep=1pt,minimum size=4mm] {$\bar r$};

% edges

\draw [-,black,thick] (pt) to node [auto] {} (pv);
\draw [-,black,thick] (pt) to node [auto] {} (pghij);

\draw [-,black,thick] (pv) to node [auto] {} (pab);
\draw [-,red,very thick,dashed] (pv) to node [auto] {} (pd);
\draw [-,red,very thick,dashed] (pv) to node [auto] {} (pu);

\draw [-,red,very thick,dashed] (pd) to node [auto] {} (pw);
\draw [-,black,thick] (pd) to node [auto] {} (pc);
\draw [-,red,very thick,dashed] (pu) to node [auto] {} (pfz);

\draw [-,black,thick] (pfz) to node [auto] {} (pe);

%

% extra edges

\draw [-,red,very thick, dashed, out=90,in=135] (w) to node [above] {$e_i=xy$} (z);

\draw [->,blue,thick,dotted,out=15,in=165] (w) to node [auto] {} (pw);
\draw [->,blue,thick,dotted,out=30,in=150] (z) to node [auto] {} (pfz);
\draw [->,blue,thick,dotted,out=-20,in=-160] (v) to node [auto] {} (pv);

\node (lca) at (0,1) [draw=blue!0,fill=blue!0,thick,inner sep=1pt,minimum size=4mm] {$\mathrm{lca}(zw) \rightarrow$};

% tree T' after contraction
\begin{scope}[xshift = -4cm, yshift=-4.5cm]

\node (ct) at (12,0) [circle,draw=blue!100,fill=blue!10,thick,inner sep=1pt,minimum size=4mm] {$\bar t$};

\node (cvdwufz) at (10,1.5) [rounded corners,draw=blue!100,fill=blue!10,thick,inner sep=1pt,minimum size=4mm] {$\bar v,\bar d,\bar x,\bar u,\bar f,\bar y$};
\node (cghij) at (14,1.5) [rounded corners,draw=blue!100,fill=blue!10,thick,inner sep=1pt,minimum size=4mm] {$\bar g,\bar h,\bar i,\bar j$};

\node (cab) at (8,3) [rounded corners,draw=blue!100,fill=blue!10,thick,inner sep=1pt,minimum size=4mm] {$\bar a,\bar b$};
\node (cc) at (10,3) [circle,draw=blue!100,fill=blue!10,thick,inner sep=1pt,minimum size=4mm] {$\bar c$};
\node (ce) at (12,3) [circle,draw=blue!100,fill=blue!10,thick,inner sep=1pt,minimum size=4mm] {$\bar r$};

\draw [-,black,thick] (ct) to node [auto] {} (cvdwufz);
\draw [-,black,thick] (ct) to node [auto] {} (cghij);

\draw [-,black,thick] (cvdwufz) to node [auto] {} (cab);
\draw [-,black,thick] (cvdwufz) to node [auto] {} (cc);
\draw [-,black,thick] (cvdwufz) to node [auto] {} (ce);

\end{scope}
% \end{footnotesize}

\end{tikzpicture}
\end{center}
\caption{Top left: The shortest path tree $T$. The next edge in $E'$ to be processed is $e_i = xy$. The cost $c_{xy}$
associated with this edge is the sum of weights of all dashed edges, with the edge $vt$ counted twice. $\mathrm{lca}(xy) = v$.
Top right: The tree $\bar T$ in the beginning of the $i$'th iteration. The labels of the super-vertices correspond to the 
sets vertices in $T$ associated with them. The dashed edges are traversed in the $i$'th iteration and contracted to
form the tree on the bottom of the figure.}\label{fig:treealg}
\end{figure}

Note that the running time of Algorithm~\ref{alg:impl} is dominated by the sorting of the costs of the
edges. This fact is a significant advantage since sorting algorithms are very efficient in practice. In fact,
if the sorted list $L$ was provided in the input, the running time of the algorithm could be improved to
$O(m + n\log n)$, via the implementation mentioned in Section~\ref{sec:k_is_one}. Another case in which 
the latter complexity bound can be attained is that of unweighted graphs, where bucketing can be used in
order to sort the costs $c_e$ in linear time. Theorem~\ref{thm:1rta} summarizes our main result.

\begin{theorem}\label{thm:1rta}
Given an instance of ORP the potential $y$ and the corresponding paths can be computed in time $O(m \log n)$.
\end{theorem}

\section{$k$-ORP}\label{sec:korp}

Let us formally define $k$-ORP, the online replacement path problem with $k$ failed edges. We refer to the
parameter $k$ as the \textit{failure parameter}.
In $k$-ORP a \textit{scenario} corresponds to a removal of any $k$ of the edges in the graph. In this setup it is no longer
convenient to describe the problem in terms of paths and detours. Instead we introduce the notion of a \textit{routing strategy}.

A routing strategy $R:2^E \times V \rightarrow V$ is a function which, given a subset $F'\subset E$ of known 
failed links and a vertex $v\in V$, returns a vertex $u\in V$. We call $E'=E\setminus F'$ the set of \textit{active edges}. 
We are assuming the existence of a certain governing mechanism, which takes as input a routing strategy and executes it on a 
given instance. Provided with a routing strategy $R$, this mechanism iteratively moves from the current vertex $u$ to the vertex 
$v = R(F',u)$, where $F'$ is the set of failed links probed so far. The process starts at a given origin $s$ with $F' = \emptyset$ 
and ends when $t$ is reached.

Since $R$ is deterministic, this process defines a unique, possibly infinite, walk $\theta_R(u,E,F)$ in $G$ for each origin $u \in V$ 
and every scenario $F \subset E$ with $|F| \leq k$.
We remark that if $G$ contains at least $k+1$ edge-disjoint $s$-$t$ paths, there exists a routing strategy, which does not cycle.
\begin{definition}\label{kdistance}
Given $E' \subset E$, a vertex $u \in V$ and a routing strategy $R$, the \textit{$k$-value} of $R$ with respect to $E'$ and $u$ is defined as
\begin{equation*}
\mathrm{Val}_k(u,E',R) = \max_{F \subseteq E',  |F| \leq k} \ell(\theta_R(u,E',F)).
\end{equation*}
The corresponding $k$-potential of $u$ with respect to $E'$ is defined as
\begin{equation*}
y^k(u,E')=\min_{R} \mathrm{Val}_k(u,E',R).
\end{equation*}
\end{definition}
Finally, $k$-ORP is to find an \textit{optimal routing strategy} $R^*$, which minimizes $\mathrm{Val}_k(s,E,R)$, namely to solve
\begin{equation*}
R^* = \argmin_{R} \mathrm{Val}_k(s,E,R).
\end{equation*}
The following relation, which is a generalization of (\ref{eq:DPequation}), gives rise to a simple
recursive algorithm for $k$-ORP. For each $v \in V$ it holds:
\begin{equation}\label{eq:kDPequation}
y^k(v,E) = \min_{v: vu\in E}\{\max\{\ell(vu)+y^k(u,E),y^{k-1}(v,E \setminus vu)\}\}.
\end{equation}
The correctness of this relation can be proved by induction on $k$,
using the arguments in Section~\ref{sec:k_is_one}. The exact statement of the algorithm is given in
Appendix~\ref{app:general_k}, alongside a complexity analysis. We summarize with the following
theorem.

\begin{theorem}\label{thm:korp}
 \sloppy $k$-ORP can be solved on undirected graphs in time  $O(m^k \log n)$. %$O( m^k + m^{k-1} n \log n )$.
\end{theorem}

\section{ORP vs. Shortest Paths}\label{sec:ORPvsSP}

As we have seen, the length of a path and its robust length are different, and often conflicting
objectives functions. In this section we mention two results relating these two objectives.

Consider first the problem of finding an optimal solution to an instance of ORP with the 
shortest possible nominal path length. This bi-objective problem asks to find a Pareto-optimal
path with respect to the latter two objective functions. Our first result asserts that this
problem can be solved in polynomial time with a simple adaptation of Dijkstra's algorithm. In fact, 
we are able to solve the following problem for every bound $B \geq OPT$.
$$
P^* = \argmin_{P\in \mathcal{P}_{s,t}: \,\, \mathrm{Val}(P) \leq B}{l(P)}.
$$
The algorithm and possible applications of this problem are given in Appendix~\ref{app:pareto}.

The second result analyzes the performance of certain greedy heuristics for $k$-ORP. We show
that a routing strategy that always tries to route along the shortest path in the remaining 
graph is a $(2^{k+1}-1)$-approximation algorithm for $k$-ORP. The details are available in Appendix~\ref{app:sp_heuristic}.

\section{Conclusions}\label{sec:conclusions}

This paper introduces a natural variant of the replacement path problem, the online replacement path problem.
ORP captures many real-life situations, which occur in faulty or large distributed networks. The most important
characteristic of ORP is that the information about the failed edge in the network is only available locally, namely
on the edge itself. 
% We define a global worst-case measure for the quality of the solution to ORP, which justifies
% optimizing both the nominal path and the replacement paths.
ORP possesses many of the nice characteristics of ordinary shortest paths. In particular, we show that ORP 
can be solved by a simple label-setting algorithm. The computational bottleneck of this algorithm is the need to
pre-compute certain shortest paths in an adapted network. We give an efficient implementations of this step, which uses
$O(m \log n)$ time and linear space. 
We also generalize the algorithm to deal with an arbitrary constant number $k$ of failed edges.
Along this vein we introduce the notion of a routing strategy. Finally, we observe that a Pareto-optimal path
with respect to ordinary distance and robust length can be found in polynomial time. 

We conclude by mentioning a number of promising directions for future research. The complexity of ORP
with a variable number $k$ of failed edges remains open. In fact, it is not clear if
the decision problem $y^k(s,E) \leq M$ is in NP. The complexity of ORP in undirected graphs may potentially be
improved to $O(m + n \log n)$. Finally, the complexity of ORP in directed graphs
remains open. Our results give an algorithm with the running time of $O(n \mathrm{SP}(n,m))$, where $\mathrm{SP}(n,m)$
is the complexity of a single shortest path computation on a graph with $n$ vertices and $m$ edges.
This problem threatens to be as challenging as RP in directed graphs.
% The complexity of ORP in directed
% graphs is completely open. This problem threatens to be as challenging as the directed replacement path problem.
% The existence of a linear space and $O(m + n \log n)$-time algorithm for undirected ORP is open. The complexity of
% finding the optimal nominal path in the case of unbounded $k$ might be an NP-hard problem. In fact it is not clear if
% the decision problem $\min_R \phi_k(R) < M$ is in NP. Finally, consider the following extension of the model, which
% might be important from the point of view of applications. Assume that every edge $e$ has a certain \textit{radius}
% $r_e$ such that if it fails, the routing mechanism is informed about this as soon as it reaches some vertex within
% $r_e$ hops from $e$. Knowing $r_e$ for every $e\in E$, devise an optimal routing strategy in the original sense (namely,
% one which minimizes to worst-case travel time).    

%bibliography
\bibliographystyle{plain}
\bibliography{lit}	
% \newpage

%%%%%%%%%%%%%%%%%%%%%%%%%%%%%%%%%%%%%%%%%%%%%%%%%%%%%%%%%%%%%%%%%%%%%%%%%%%%%%%%%%%%%%%%%%%%%%%%%%%%%%%%%%%%%%%%%%%%%%%%%%%%
%%%%%%%%%%%%%%%%%%%%%%%%%%%%%%%%%%%%%%%%%%%%%%%%%%%%%%%%%%%%%%%%%%%%%%%%%%%%%%%%%%%%%%%%%%%%%%%%%%%%%%%%%%%%%%%%%%%%%%%%%%%%
% appendices

\appendix

\section{Proofs}\label{app:proofs}

\subsection{Proof of Lemma~\ref{lem:Tproperty}}
Applying the definition we compute
\begin{equation*}
\begin{array}{lll}
\ell(vu)+\mathrm{Val}(P_u)&=& \ell(vu)+\max\{\ell(P_u),\max_{zw\in P_u}\{\ell(P_u[u,z])+ s^{-zw}_{z} \}\} \\
&=& \max\{\ell(vu)+\ell(P_u),\max_{zw\in P_u}\{\ell(vu)+ \ell(P_u[u,z])+ s^{-zw}_{z} \}\}\\
% &=& \max\{\ell(P_v),\max_{zw\in P_u}\{\ell(vu)+\ell(P_v[u,z])+ s^{-zw}_{z} \}\}\\
&=& \max\{\ell(P_v),\max_{zw\in P_v \setminus \{vu\}}\{\ell(P_v[v,z])+ s^{-zw}_{z} \}\}.
\end{array}
\end{equation*}

Substituting in the desired expression we obtain

\begin{equation*}
\begin{array}{lll}
\max\{\ell(vu)+\mathrm{Val}(P_u),s^{-vu}_{v} \} = \\ 
\max\{\ell(P_v),\max_{zw\in P_v \setminus \{vu\}}\{\ell(P_v[u,z])+ s^{-zw}_{z} \}, s^{-vu}_{v} \} = \mathrm{Val}(P_v),
\end{array}
\end{equation*}
which proves the lemma.

\subsection{Proof of Lemma~\ref{lem:dp_operator}}
Assume towards contradiction that there exists a path $P^*_v\in \mathcal{P}_{v,t}$ such that $\mathrm{Val}(P^*_v)<\mathrm{Val}(P_v)$. 
Let $w\in U$ and $z\in V\setminus U$ be two vertices such that $zw \in P^*_v$. Consider the partition of $P^*_v$ given by 
$P^*_v = P^*_v[v,z] \cup \{zw\} \cup P^*_v[w,t]$ (Note that if $w=t$ then $P^*_v[w,t] = \emptyset$ and if 
$z=v$ then $P^*_v[v,z] = \emptyset$).

By the choice of $vu$ we have
\begin{equation*}\label{uno}
\max\{\ell(vu)+y(u),s^{-vu}_{v}\} \leq \max\{\ell(zw)+y(w),s^{-zw}_{z}\},
\end{equation*}
which, by $\mathrm{Val}(P_u) = y(u)$ and $\mathrm{Val}(P^*_v[w,t]) \geq y(w)$ implies
\begin{equation*}\label{due}
\max\{\ell(vu)+\mathrm{Val}(P_u),s^{-vu}_{v}\} \leq \max\{\ell(zw)+\mathrm{Val}(P^*_v[w,t]),s^{-zw}_{z}\}.
\end{equation*}

The latter inequality and Lemma~\ref{lem:Tproperty} give $\mathrm{Val}(P_v) \leq \mathrm{Val}(P^*_v[z,t])$. On the other hand, by monotonicity we
have $\mathrm{Val}(P^*_v) \geq \mathrm{Val}(P^*_v[z,t])$. We conclude that $\mathrm{Val}(P_v) \leq \mathrm{Val}(P^*_v[z,t]) \leq \mathrm{Val}(P^*_v)$; a contradiction.

\subsection{Proof of Lemma~\ref{lem:expression_svalues}}
Clearly this choice represents a $u$-$t$ path
in $G-uu'$, so $s^{-uu'}_u \leq d_G(u,v) + \ell(vw) + d^*(w)$ holds. 
Assume towards contradiction that a better $u$-$t$ path $P$ existed in $G-uu'$, namely 
$\ell(P) < d_G(u,v) + \ell(vw) + d^*(w)$. 
To reach $t$ from $u$ without using $uu'$ the path $P$ needs to contain an edge $v'w'$, such that $v'\in T_u$ and 
$w'\not\in T_u$. We obtain
\begin{equation*}
\begin{array}{lll}
\ell(P) = \ell(P[u,v']) + \ell(v'w') + \ell(P[w',t]) &\geq&  d_G[u,v'] + \ell(v'w') + d^*(w') \\
 &\geq&  d_G(u,v) + \ell(vw) + d^*(w),
\end{array}
\end{equation*}
which contradicts the choice of $P$. 

\newpage

\section{A Summary of the Algorithm in Section~\ref{sec:impl}}\label{svaluesalg}

\begin{algorithm}
\caption{ }\label{alg:impl}
\begin{algorithmic}[1]
\STATE{Compute the shortest path tree $T$.}
\STATE{Create a copy $\bar T$ of $T$.}
\STATE{Store a pointer $p[u]$ in every vertex $u$ in $T$ to its corresponding vertex $\bar u$ in $\bar T$.}
\STATE{Store a pointer $p[\bar u]$ in every vertex $\bar u$ in $\bar T$ to its corresponding vertex $u$ in $T$.}
\STATE{Store similar pointers $p[\bar e]$ from edges in $\bar T$ to corresponding edges in $T$.}
\STATE{Compute $c_e$ for every $e\in E'$.}
\STATE{Sort $E'$ according to increasing order of $c_e$. Let $L$ be the sorted list.}
\WHILE{$\bar T$ contains more than one vertex}
  \STATE{Remove the first edge $e = xy$ from $L$.}        
  \STATE{$\bar x = p[x]$, \quad $\bar y = p[y]$.}
  \STATE{$z = \mathrm{lca}(e)$, \quad $\bar z = p[\mathrm{lca}(e)]$.}
  \STATE{Let $Y \subset \bar T$ denote the union of the $\bar x$-$\bar z$ and $\bar y$-$\bar z$ paths in $\bar T$.}
  \FOR{$\bar a\in Y$}
    \STATE{$uu' = p[\bar a]$.}
    \STATE{$s^{-uu'}_u = c_{e} - d^*(u)$.}
  \ENDFOR
  \STATE{Contract $Y$ in $\bar T$ and update pointers.}
  \REPEAT
    \STATE{Let $e=vw$ be the first edge in $L$.}
    \IF{$p[v] = p[w]$}
      \STATE{Remove $e$ from $L$.}
    \ENDIF
  \UNTIL{$p[v] \neq p[w]$ \OR $L = \emptyset$.}
\ENDWHILE
\end{algorithmic}
\end{algorithm}

\section{An Algorithm for $k$-ORP}\label{app:general_k}

Algorithm~\ref{alg:korp} solves (\ref{eq:kDPequation}) for every $v\in V-t$ and computes the corresponding
optimal routing strategy. Note that Algorithm~\ref{alg:korp} is identical to Algorithm~\ref{alg:main}
when $k=1$.
To formally prove the correctness of Algorithm~\ref{alg:korp} one needs to state equivalent monotonicity properties, as well
as analogues of Lemma~\ref{lem:Tproperty} and Lemma~\ref{lem:dp_operator}. They are however omitted, as they are identical to 
those of Section~\ref{sec:k_is_one}.

\begin{algorithm}
\caption{ }\label{alg:korp}
\begin{algorithmic}[1]
\STATE{Compute $y^{k-1}(u,E \setminus uv)$ for each $uv \in E$.}
\STATE{$U = \emptyset$.}
\STATE{$W = V$.}
\STATE{$y^{k}(t,E) = 0$.}
\STATE{$y^{k}(u,E) = \infty$ for each $u \in W$.}
% \STATE{$y^{k}(u,E):=\infty$ for each $u \in V$.}
\WHILE{$U \neq V$}
 
      \STATE{Find $u = \argmin_{z \in W}\{y^{k}(z,E)\}$. }
      \STATE{$U = U + u$.}
      \STATE{$W = W - u$.}
      \FORALL{ $vu \in E$ such that $v \in W$}
      	\IF{ $y^{k}(v,E) > \max\{\ell(uv)+y^{k}(u,E),y^{k-1}(v,E\setminus \{vu\})\}$ }
      		\STATE{$y^{k}(v,E) = \max\{\ell(uv)+y^{k}(u,E),y^{k-1}(v,E\setminus \{vu\})\}$.}
%       		\STATE{$u=R(\emptyset,v)$.}
				\ENDIF
			\ENDFOR
			
\ENDWHILE
\end{algorithmic}
\end{algorithm}

We conclude by bounding the complexity of a naive implementation of this algorithm. Let $T(m,n,k)$ denote the running time 
of the algorithm on a graph with $n$ vertices and $m$ edges and failure parameter $k$. The algorithm in Section~\ref{sec:impl} gives 
$T(n,m,1) = O(m \log n)$. For $k>1$ we have 
% $T(m,n,k) = O(m + n \log n) + m T(m-1,n,k-1)=O(m + n \log n) + {m \choose k-1}T(n,m,1)$. 
$T(m,n,k) = O(m + n \log n) + m T(m-1,n,k-1) = O(m^{k-1}T(n,m,1))$. 
% This finally gives $T(n,m,k) = O(m^k + m^{k-1} n \log n)$. 
\sloppy This finally gives $T(n,m,k) = O(m^k \log n)$. 
% We summarize the result of this section in the following theorem.

\section{Pareto-Optimal ORPs}\label{app:pareto}

In this section we are concerned with obtaining a path $P^*$ with robust length at most $B$, and
a shortest nominal path length among all such paths. 
This problem is equivalent to that of finding a Pareto-optimal path with respect to the objective functions
corresponding to the ordinary distance function, and the robust length. 
We assume $B\geq OPT$, the
optimal solution value of the corresponding ORP instance.
Formally, we aim at finding an $s$-$t$ path $P^*$ satisfying

$$
P^* = \argmin_{P\in \mathcal{P}_{s,t} : \,\, \mathrm{Val}(P)\leq B}{l(P)}.
$$

The algorithm for this problem is a slightly modified version of Dijkstra's algorithm on the directed edge-weighted 
graph $G'=(V,K, \omega)$, where $K$ has two directed edges $uv$ and $vu$ for each undirected edge $uv \in E$, with $\om(uv)=\om(vu)=\ell(uv)$.
Algorithm~\ref{alg:pareto} is a formal statement of the algorithm. Note that the only difference with Dijkstra's algorithm 
is the condition in step $8$. Unlike ordinary shortest paths, when a vertex $u$ is selected, and the distance labels of its neighbors
are updated, it is not sufficient to check for each $uv \in K$ the usual condition
$$d(u)+ \om(uv) \leq d(v).$$
We need to additionally verify whether the edge $uv$ can belong to a path with robust length
$B$ or not. In other words, we need to check whether the length of the path from $s$ to $u$ plus the length of the detour from $u$ to $t$ 
avoiding $uv$ is at most $B$, namely 
$$d(u) + s^{-uv}_{u} \leq B.$$
Algorithm~\ref{alg:pareto} can clearly be implemented to run in $O(m + n \log n)$ time using the results
of Section~\ref{sec:impl}.

\begin{algorithm}
\caption{}\label{alg:pareto}
\begin{algorithmic}[1]
\STATE{ $S=\emptyset$;  \quad $\bar S = V$}
\STATE{ $d(s) = 0$; \quad $d(v) = \infty \,\, \forall v \in V - s$} %; \quad $p(v)=\text{NIL} \,\, \forall v \in V$}
%\STATE{Compute $s^{uv}_{ut}$ for each $uv \in E$}
\WHILE{$t \notin S$}
 			\STATE{Find $u = \argmin_{z \in \bar S}{d(z)}$}
 			\STATE{$S = S + u$}
 			\STATE{$\bar S = \bar S - u$}
      \FOR{$v \in N(u)\setminus S$}
      			 \IF{$d(u)+ \om(uv) \leq d(v)$ and $d(u) + s^{-uv}_{u} \leq B$}
      				 \STATE{$d(v)=d(u)+ \om(uv)$}
%       				 \STATE{$p(v)=u$}
			  \ENDIF
      \ENDFOR
\ENDWHILE
\end{algorithmic}
\end{algorithm}

Let us briefly discuss the potential applications of the latter problem. While shortest paths
often have undesirable behavior in unreliable networks, an optimal solution to ORP might have
a prohibitively large cost. In some applications faults occur rarely, hence it is preferred to
have the cost of the nominal path as low as possible. At the same time, it is necessary that
the cost does not exceed a certain threshold $B$, in every scenario. Consequently, it makes sense
to regard the threshold $B$ as a hard constraint on the robust length, and optimize the length
of the nominal path. Algorithm~\ref{alg:pareto} gives the decision maker the desired freedom
to choose the level of conservatism that she desires.

\section{The Shortest Path Heuristics}\label{app:sp_heuristic}

It is common to use heuristics which rely on shortest paths in routing algorithms. In this section we show that a
naive shortest path routing strategy performs very poorly for $k$-ORP. In fact, the approximation guarantee it provides
grows exponentially with the adversarial budget $k$.
To this end we define more formally the shortest path heuristics, which we denote by 
$R_{SP}$. The routing strategy $R_{SP}$ works as follows. At each vertex $u\in V$ and given a set
of known failed edges $F'$, $R_{SP}$ tries to route the package along a shortest path in the remaining graph
$G-F'$. In the following lemma we show
that $R_{SP}$ is a factor $2^{k+1}-1$ approximation for the optimal routing strategy, in the presence of at most $k$
failed edges.

\begin{lemma}\label{lem:sp_heuristic_general}
 Let $\mathcal{I} = (G,s,t)$ be an instance of $k$-ORP. Then
$$
\mathrm{Val}_k(s,E,R_{SP}) \leq (2^{k+1}-1) y^k(s,E).
$$
\end{lemma}

\begin{proof}
 Let $OPT = y^k(s,E)$, $R^* = \argmin_R {\mathrm{Val}_k(s,E,R)}$ and $Q^0 = \theta_{R^*}(s,E, \emptyset)$ be the corresponding
nominal path. Consider a set $F$ of failed edges with $|F| \leq k$. Define $\omega = \theta_{R_{SP}}(s,E,F)$ to be
the walk followed by $R_{SP}$ in $G-F$, and let $\omega = (s = u_1, u_2, \cdots, u_m = t)$ be the corresponding sequence of vertices.
% The routing strategy $R_{SP}$ defines a certain walk $\omega = (s = u_1, u_2, \cdots, u_m = t)$ in $G-F$. 
We divide our analysis according to the
number of failed edges encountered by $R_{SP}$.
We prove by induction on $i$, that if the routing strategy encountered a total of $i\leq k$ failed edges then 
\begin{equation}\label{eq:induction}
\ell(\omega) \leq (2^{i+1} - 1) OPT. 
\end{equation}
The base case $i=0$ corresponds to scenarios in which no failed edge is encountered in the routing.
In this case $\omega$ is simply a shortest $s$-$t$ path in $G$, hence $\ell(\omega) \leq OPT$, as required.
Assume next that (\ref{eq:induction}) holds for every $j<i$ and consider the case that the routing encounters
exactly $i$ failed edges. We can assume without loss of generality that all edges of $F$ were probed and
$|F| =i$.

Let $r < m$ be such that $u_r$ is the vertex incident to the $i$'th failed edge
in the routing. In other words, before reaching $u_r$, the routing probed exactly $i-1$ failed edges. Let 
$e = \{u_r,w\}$ be the $i$'th failed edge probed by the routing strategy. 
Consider an execution of the routing strategy on the same instance with the different failure scenario 
$F' = F-e$. The resulting walk $\omega'$ will have the first $r$ vertices in common with $\omega$,
namely the sub-walk $\sigma = (s=u_1,\cdots, u_r)$ will appear in both walks. By the inductive hypothesis
we have that $\ell(\sigma) \leq \ell(\omega') \leq (2^i-1)OPT$. It remains to bound the length of the tail of
$\omega$ from $u_r$ until $u_m=t$ to complete the proof. To this end recall that $R_{SP}$ routes the
package along the shortest remaining path in the graph. Since the last failure encountered by $R_{SP}$
is $e$, the remaining path is simply the shortest $u_r$-$t$ path in $G-F$. To bound the length of this
path we construct a $u_r$-$t$ walk $\theta$ as follows. First $\theta$ traces the entire route taken by
$R_{SP}$ back to $s$ and then uses the walk that $R^*$ would use to reach $t$ from $s$ in the scenario
$F$. Clearly we have $\ell(\theta) \leq (2^i-1)OPT + OPT$. Furthermore, this walk is intact in $G-F$. This
gives the required bound $\ell(\omega) \leq (2^i-1) OPT + (2^i-1)OPT + OPT = (2^{i+1}-1)OPT$
and finishes the proof. \qed
\end{proof}

The bound obtained in Lemma~\ref{lem:sp_heuristic_general} seems crude at first glance. In particular,
in the inductive step we follow the entire walk performed so far backwards to reach $s$ and start over.
In the following example we show that the bound of Lemma~\ref{lem:sp_heuristic_general} is tight.

\begin{example}\label{ex:general_sp_bad}
 Let $M \in \mathbb{Z}_+$ be a large integer. Consider the following instance $\mathcal{I} = (G,s,t)$ of $k$-ORP.
The graph contains $k+1$ parallel edges connecting $s$ and $t$ with length $M+1$.
In addition the graph contains a path $(s,u_1, \cdots, u_{k})$ of length $k+1$. The edge $su_1$ has length $M$ 
and every edge $u_iu_{i+1}$ has length $2^iM$. Finally, the vertices $u_1, \cdots, u_k$ are connected to $t$ with
edges of length zero. The construction is illustrated in Figure~\ref{fig:sp_heuristic_bad}.

Consider the failure scenario, which fails all edges $u_it$ for $i\in [k]$. The routing strategy $R_{SP}$ will
follow the path $(s, u_1, \cdots, u_k)$, then follow it back to $s$ and then take one of the edges with length $M+1$
to $t$. The total length of this walk is $2(M + 2M + \cdots + 2^k M) + M + 1 = (2^{k+1} - 1)M + 1$. At the same
time the optimal routing strategy routs the package along the edges with length $M+1$ with a worst-case cost of $M+1$.
The ratio between the two numbers tends to $2^{k+1} -1$ as $M$ tends to infinity. 
\end{example}

\begin{figure}[h]
\begin{center}
\begin{tikzpicture}[scale=0.8]

%main graph
\node (s) at (0,0) [circle,draw=blue!100,fill=blue!10,thick,inner sep=1pt,minimum size=6mm] {$s$};
\node (t) at (6,0) [circle,draw=blue!100,fill=blue!10,thick,inner sep=1pt,minimum size=6mm] {$t$};

\node (u1) at (3,0) [circle,draw=black!100,fill=black!100,thick,inner sep=1pt,minimum size=2mm] {};
\node [label={[label distance=0mm] 135: $u_1$}] at (u1) {};
\node (u2) at (4,1.8) [circle,draw=black!100,fill=black!100,thick,inner sep=1pt,minimum size=2mm] {};
\node [label={[label distance=0mm] 135: $u_2$}] at (u2) {};
\node (u3) at (6,3) [circle,draw=black!100,fill=black!100,thick,inner sep=1pt,minimum size=2mm] {};
\node [label={[label distance=0mm] 45: $u_3$}] at (u3) {};

\node (ui) at (8.2,1.5) [circle,draw=black!100,fill=black!100,thick,inner sep=1pt,minimum size=2mm] {};
\node [label={[label distance=0mm] 45: $u_i$}] at (ui) {};
\node (ui1) at (9,0) [circle,draw=black!100,fill=black!100,thick,inner sep=1pt,minimum size=2mm] {};
\node [label={[label distance=0mm] 0: $u_{i+1}$}] at (ui1) {};

\node (uk) at (7.5,-1.8) [circle,draw=black!100,fill=black!100,thick,inner sep=1pt,minimum size=2mm] {};
\node [label={[label distance=0mm] 0: $u_{k-1}$}] at (uk) {};
\node (uk1) at (6,-2.6) [circle,draw=black!100,fill=black!100,thick,inner sep=1pt,minimum size=2mm] {};
\node [label={[label distance=0mm] 180: $u_{k}$}] at (uk1) {};

% edges

\draw [-,black,thick] (u3) to node [auto] {} (ui);
\draw [-,black,thick] (ui1) to node [auto] {} (uk);

\node (dots1) at (7.2,2.3) [circle,draw=blue!0,fill=blue!0,thick,inner sep=1pt,minimum size=12mm] {$\ddots$};
\node (dots1) at (8.2,-0.8) [circle,draw=blue!0,fill=blue!0,thick,inner sep=1pt,minimum size=12mm] {$\iddots$};

\draw [-,black,thick] (s) to node [auto] {$M$} (u1);
\draw [-,black,thick,dashed] (u1) to node [auto] {$0$} (t);
\draw [-,black,thick] (u1) to node [auto] {$2M$} (u2);
\draw [-,black,thick] (u2) to node [auto] {$4M$} (u3);
\draw [-,black,thick,dashed] (u2) to node [auto] {$0$} (t);
\draw [-,black,thick,dashed] (u3) to node [auto] {$0$} (t);
\draw [-,black,thick] (ui) to node [auto] {$2^i M$} (ui1);
\draw [-,black,thick,dashed] (ui) to node [auto] {$0$} (t);
\draw [-,black,thick,dashed] (ui1) to node [auto] {$0$} (t);
\draw [-,black,thick] (uk) to node [auto] {$2^{k-1} M$} (uk1);
\draw [-,black,thick,dashed] (uk) to node [auto] {$0$} (t);
\draw [-,black,thick,dashed] (uk1) to node [auto] {$0$} (t);

\node (dummy) at (3,-1.6) [circle,draw=blue!0,fill=blue!0,thick,inner sep=1pt,minimum size=12mm] {$\cdots$};
\draw [-,black,very thick,out=-30,in=210] (s) to node [auto] {$M+1$} (t);
\draw [-,black,very thick,out=-45,in=225] (s) to node [auto] {} (t);
\node (dummy) at (3,-2.4) [circle,draw=blue!0,fill=blue!0,thick,inner sep=1pt,minimum size=12mm] {$M+1$};
\draw [-,black,very thick,out=-90,in=270] (s) to node [auto] {} (t);

\end{tikzpicture}
\end{center}
\caption{A bad example for the shortest path heuristic. The dashed edges correspond to the worst case scenario.}\label{fig:sp_heuristic_bad}
\end{figure}
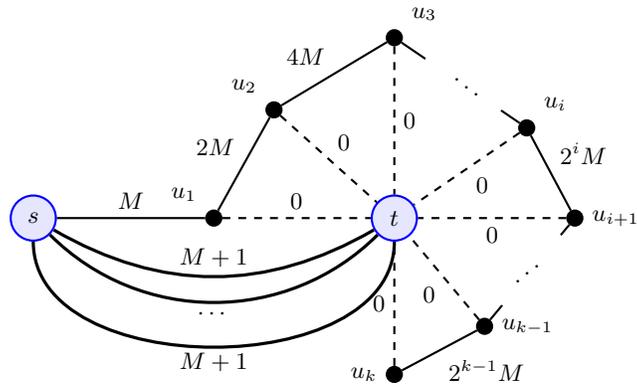

% \vspace{20cm}

% hack to put picture on top
\begin{figure}[h]
\begin{center}
\begin{tikzpicture}[scale=0.8]

\draw [-,white] (0,0) to (0,13);

\end{tikzpicture}
\end{center}
% \caption{A bad example for the shortest path heuristic. The dashed edges correspond to the worst case scenario.}\label{fig:sp_heuristic_bad}
\end{figure}

\end{document}